\documentclass[11pt]{article}

\usepackage[alg]{KLMM/klmm}
\usepackage{xcolor}

\newcommand{\Z}{\mathbb Z}

\def\R{{\mathbb{R}}}

\def\Z{{\mathbb{Z}}}

\def\R{{\mathcal{R}}}
\def\K{{\mathcal{K}}}

\def\F{{\mathbb{F}}}

\def\Pr{{\mathbb{P}}}

\newcommand{\NP}{\mathsf{NP}}
\newcommand{\BPP}{\mathsf{BPP}}
\newcommand{\DIV}{\mathsf{DIV}}
\newcommand{\NONDIV}{\mathsf{NONDIV}}
\newcommand{\SparseRoot}{\mathsf{SparseRoot}}

\DeclareMathOperator{\cont}{cont}
\DeclareMathOperator{\lc}{lc}

\title{Sparse Polynomial Divisibility Test over Finite Field is CoNP-hard}
\author{\fontsize{11}{16}\selectfont
    Yichuan Cao\institution{State Key Laboratory of Mathematical Sciences, Academy of Mathematics and Systems Science, Chinese Academy of Sciences; University of Chinese Academy of Sciences}\equalcontrib
    \and Ruichen Qiu\instref{1}\equalcontrib
    \and Qiao-Long Huang\institution{Shandong University}
    \and Ruyong Feng\instref{1}
    \and Xiao-Shan Gao\instref{1}\Cauthor
}
\date{}

\begin{document}
\maketitle

\begin{abstract}
\noindent
In this paper, we show that deciding whether a sparse polynomial does not divide another sparse polynomial exactly over finite fields is NP-hard under BPP many-one reductions. Equivalently, the sparse polynomial divisibility test over finite fields is CoNP-hard. This resolves the long-standing open problem concerning the computational complexity of the divisibility test for sparse polynomials in the setting of finite fields.
\end{abstract}

\section{Introduction}

Polynomial arithmetic is a fundamental algebraic operation that appears in many scientific and engineering computations and is a basic issue in computer algebra and the theory of computation~\cite{ModernCA,Jin2025STOC,NickFischer2025SODA}.
%
%
For two dense degree-$D$ univariate polynomials over a field, the fast Euclidean algorithm uses $O(D\,\log^2 D\,\log\log D)$ field operations~\cite[Chapter 11]{ModernCA}. 
Thus, the divisibility test for two dense polynomials can be done with the same complexity.

However, polynomials are typically expressed in sparse form~\cite{SparsityChallenges}, meaning they are given as a list of their nonzero coefficients together with the associated exponents.
Sparse polynomials are also the standard form for polynomials in computer algebra systems such as Maple and Mathematica.
%
%
For a sparse univariate polynomial $f$ of degree $D$, containing $T$ nonzero terms, and whose coefficients have a maximum absolute value of $H$, its size is tightly bounded by
$$\text{size}(f) = O(T(\log D+\log H))=O(T\log(DH)).$$
So a sparse polynomial-time algorithm should have bit complexity that is polynomial in $T$, $\log D$, and $\log H$.

Arithmetic for sparse polynomials is quite straightforward. However, some seemingly simple problems are shown to be NP-hard.
Deciding whether $x^N - 1$ cannot divide the product of a finite set of sparse polynomials over the integers is NP-complete~\cite{plaisted1984}.
Assuming the Extended Riemann Hypothesis (ERH), the divisibility test problem belongs to coNP~\cite{Grigoriev-ISSAC1992}.
Deciding whether a sparse polynomial is divisible by the square of a degree one polynomial in $\F_p[x]$ is NP-hard \cite{BI-SIAMJC2016}.

The intrinsic difficulty of the divisibility test is that using the Euclidean algorithm to check divisibility is not polynomial-time.
When performing long division on univariate sparse polynomials, the size of the quotient can range from constant to exponential in the input size, both in terms of the number of nonzero monomials and the sizes of the coefficients, as illustrated by the following example from~\cite{Grenet-H2026}.

\begin{example} 
Let $f = x^{2d+1} - x^{2d}$ and $g = x^{d+1} - x^d + 1$ in $\mathbb{Z}[x]$. Then, the quotient is $x^d + \sum_{i=0}^{d-1} 2^{d-i-1}x^i$ and the remainder is $(2^d - 3)x^d - \sum_{i=0}^{d-1} 2^{d-i-1}x^i$ when performing the long division of $f$ with respect to $g$.
\end{example}
Even if $g$ divides $f$, this difficulty still exists, as illustrated by the following simple example.
\begin{example}
Let $f=x^D-1$ and $g=x-1$. Then $f/g=\sum_{i=0}^{D-1} x^i$, whose term is exponential in the size of $f,g$. 
\end{example}

The divisibility test for sparse polynomials is a natural and long-standing open problem, explicitly proposed in \cite{Karpinski1999CH},
\cite[Challenge 3]{SparsityChallenges}, 
and \cite[Open Problem 3]{Roche2018}. 
More recently, \cite[Open Problem 9.4 ]{Grenet-H2026} states the problem more precisely as:
%
%
%

\begin{quote}
\textbf{Open Problem 1.} 
Given two sparse polynomials $f, g \in \K[x]$ over a field $\K$, what is the complexity of testing whether $g$ divides $f$? Is it possible to perform the computation in polynomial-time in the sparse representation of $f$ and $g$? Is the problem NP-hard?
\end{quote}

In this paper, we consider divisibility test over finite fields, and we restate Open Problem 1 in a more explicit form below.

\begin{quote}
    \textbf{Open Problem 2.} 
    Given two sparse polynomials $f,g \in \F_q[x]$, is it possible to 
    determine whether $g$ divides $f$ exactly in polynomial-time in the sparse representation of $f$ and $g$?
\end{quote}


\subsection{Main Results}
In this paper, we settle Open Problem 2 by showing that 
the sparse divisibility test problem over finite fields is CoNP-hard. 
\begin{theorem}
\label{thm-main}
Given two sparse polynomials $f,g \in \F_q[x]$, deciding whether $g$ does not divide $f$ exactly is NP-hard under BPP many-one reductions. 
Consequently, under the standard complexity assumption \(\mathrm{NP}\nsubseteq\mathrm{BPP}\), no polynomial-time algorithms exist for the non-divisibility test for sparse polynomials with finite field coefficients.
\end{theorem}
%
The proof constructs a randomized polynomial reduction from an NP problem to the non-divisibility test problem using a new NP-hard problem from~\cite{BI-SIAMJC2016}.
Observe that Theorem \ref{thm-F} establishes the result in the case of $q$ being a prime number, and this is sufficient to deduce Theorem \ref{thm-main}.

Observe that Theorem \ref{thm-main} applies to multivariate polynomials over finite fields, with the univariate situation arising as a particular special case. Further, the statement  remains valid for any field $\K$ of positive characteristic $p$, because $\F_p$ is a subfield of $\K$.

In the finite field setting, Theorem \ref{thm-main} shows that there exists no randomized polynomial-time algorithm for the non-divisibility test under standard complexity assumptions.
This raises several new questions: Is non-divisibility NP-complete under BPP reductions? Is non-divisibility NP-complete in the usual sense? Or is non-divisibility actually PSPACE-complete?
In the rational numbers setting, there exists no essential progress for Open Problem 1. 


The main results presented in this paper are obtained through an interactive collaboration between the authors and an artificial intelligence agent system, \textit{MechMath Agent Team (MMAT)}~\cite{MMAT}. 
The authors assume full responsibility for the paper’s content.
MMAT is a large language model driven agent designed to prove mathematical theorems expressed in both natural language and formal language in Lean. 

\subsection{Lean Formalization}
The algebraic core of our finite-field hardness proof, the Divisibility Bridge of Lemma~\ref{lem:bridge}, has been formally verified in Lean~4 with version \texttt{4.29.0}~\footnote{\url{https://github.com/EonMath/sparse-div}.}.  

Modeling \(\F_p\) as \texttt{ZMod\,p} with \(p\) prime and
``\(P\) has a root in \(\F_p\)'' as \(\exists a\in\F_p,\ P(a)=0\), the
formalization defines the map \(P\mapsto(f_P,g_P)\), including the guarded case
\(P=0\), and proves all of its asserted properties: that \(g_P\ne0\); the
equivalence \(g_P\mid f_P\Leftrightarrow P\) has no root in \(\F_p\), together
with its contrapositive root form
\(g_P\nmid f_P\Leftrightarrow P\) has a root in \(\F_p\); and the degree bounds
\(\deg f_P\le pE\) and \(\deg g_P\le p+E\).  The two key steps of the proof are
isolated as reusable lemmas, namely the characteristic-\(p\) Frobenius identity
\(P(x^p)=P(x)^p\) and the divisibility criterion
\((x^p-x)\mid R\Leftrightarrow R\) vanishes on all of \(\F_p\).  The development
contains no \texttt{sorry} and depends only on \texttt{mathlib}'s standard
axioms.  The remaining ``at most \(2t\) terms'' and sparse bit-size blowup
clauses of Lemma~\ref{lem:bridge} are informal statements about the sparse
encoding and are not part of the formal development.

We do \emph{not} formalize the hardness theorem itself
(Theorem~\ref{thm-F}).  Its statement quantifies over the complexity classes
\(\NP\) and \(\BPP\) and asserts NP-hardness \emph{under BPP many-one
reductions}; verifying it formally would first require a substantial Lean theory
of computational complexity---a machine model, polynomial-time randomized
reductions, and the classes \(\NP\) and \(\BPP\)---none of which is currently
available in \texttt{mathlib}.  What the reduction actually uses beyond this
complexity-theoretic scaffolding is precisely the algebraic equivalence of
Lemma~\ref{lem:bridge}: given the NP-hardness of \(\SparseRoot\)
(Theorem~\ref{thm:bcr}), the bridge turns a root-existence query into a single
non-divisibility query, while the surrounding wrapper is a standard many-one
reduction.  We therefore formalize that algebraic core and leave the
complexity-class bookkeeping at the paper level.

\subsection{Related Work}

For dense degree-$D$ polynomials $f,g$ over a field. The fast Euclidean algorithm computes the quotient and remainder $h,r$ such that $f = hg + r$ using $O(D\,\log^2 D\,\log\log D)$ field operations~\cite[Chapter 11]{ModernCA}. Consequently, a divisibility test can be performed within the same complexity bound.

For sparse polynomials, a polynomial-time algorithm is provided for the divisibility test in the case where $g$ has a specific special form~\cite{Giorgi2021-ISSAC}. 
There are numerous known results for computing the quotient under the assumption of exact divisibility.
Let $f, g \in \R[x]$ be such that $g$ divides $f$, $D = \deg(f)$,  $T = \max(|f|_0, |g|_0, |f/g|_0)$, and  $H = \max(|f|_\infty, |g|_\infty, |f/g|_\infty)$. The polynomial $f/g$ can be computed with $\widetilde{O}(T\log(DH))$ bit operations if $R = \mathbb{Z}$;
$\widetilde{O}(T \log D \log q)$ bit operations if $R = \mathbb{F}_q$ has characteristic $> D - \deg(g)$;
and $\widetilde{O}(T \log^2(D)\log(Dq))$ bit operations if $R = \mathbb{F}_q$ has characteristic $\leq D - \deg(g)$~\cite{Giorgi2022-ISSAC,Giorgi2021-ISSAC}.
An algorithm for computing $f/g$ with bit complexity $\widetilde{O}(T(\log H + T \log D))$ is also provided in~\cite{Nahshon-Sigma2026}.

On the computational complexity of the sparse divisibility test, deciding whether $x^N - 1$ does not divide the product of a finite set of sparse polynomials over integers is 
NP-complete~\cite{plaisted1984,plaisted1978some,plaisted1977sparse}.
For two sparse polynomials $f(x)$ and $g(x)$ with integer coefficients, the following problems are NP-hard: determining if the quotient when $f(x)$ is divided by $g(x)$ has a non-zero constant term and determining the degree of the remainder when $f(x)$ is divided by $g(x)$~\cite{plaisted1984}.
Assuming the Extended Riemann Hypothesis (ERH), the divisibility problem over integers belongs to coNP~\cite{Grigoriev-ISSAC1992}.
Note that the result in~\cite{Grigoriev-ISSAC1992}  only establishes an upper bound (membership in CoNP) and does not claim that the nondivisibility problem is hard.
Deciding whether a sparse polynomial is divisible by the square of a degree one polynomial in $\F_p[x]$ is NP-hard \cite[Theorem 1.5]{BI-SIAMJC2016}.
This result demonstrates that there is a linear polynomial $l$ for which deciding whether $l^2 \mid f$ holds is hard; hence, it does not imply hardness for the general divisibility problem.


The rest of the paper is organized as follows.
In Section \ref{sec-F}, we show that the divisibility test is CoNP-hard over finite fields.
In Section \ref{sec-red}, we present an alternative proof of Theorem \ref{thm-main} by constructing a polynomial-time reduction.
In Section \ref{sec-conc}, conclusions are presented.
In the appendix, we present a randomized polynomial-time reduction from the non-divisibility problem over the integers to the non-divisibility problem over finite fields.
We present this fact along with its proof because the reduction may be interesting in its own right and could be useful in other settings.

\section{Non-divisibility Test over Finite Fields is NP-hard}
\label{sec-F}
Sparse polynomials are represented by coefficient-exponent lists, with
exponents written in binary.  On the nonzero-divisor domain over prime fields,
let
\[
  \DIV=\{(p,f,g): p\text{ prime},\ f,g\in \F_p[x],\ g\ne0,\ g\mid f\},
\]
and let \(\NONDIV\) be its complement on the same well-formed domain.  The
convention for zero divisors is irrelevant below since every constructed
target instance has \(g\ne0\).

\subsection{The Divisibility Bridge}

\begin{lemma}
\label{lem:bridge}
Let \(p\) be prime and let \(P\in\F_p[x]\) be sparse.  If \(P\ne0\), define
\[
  f_P=P(x^p)-P(x),
  \qquad
  g_P=(x^p-x)P(x).
\]
If \(P=0\), define instead the guarded pair
\[
  f_P=1,\qquad g_P=x.
\]
Then \(g_P\ne0\) and
\[
  g_P\mid f_P
  \quad\Longleftrightarrow\quad
  P\text{ has no root in }\F_p.
\]
Equivalently,
\[
  g_P\nmid f_P
  \quad\Longleftrightarrow\quad
  P\text{ has a root in }\F_p.
\]
Moreover, if \(P\ne0\) has \(t\) nonzero terms and degree \(E\), then each of
\(f_P\) and \(g_P\) has at most \(2t\) terms after collection,
\[
  \deg f_P\le pE,\qquad \deg g_P\le p+E,
\]
with the convention that \(\deg 0=-\infty\).  In the guarded case \(P=0\),
both outputs have one term and degrees \(0\) and \(1\), respectively. 
Thus, the map \(P\mapsto(f_P,g_P)\) has polynomial sparse bit-size blowup.
\end{lemma}

\begin{proof}
The guarded case is immediate: the zero polynomial has every element of
\(\F_p\) as a root, while \(x\nmid 1\), and the divisor \(x\) is nonzero.

Assume now that \(P\ne0\).  Put \(A=x^p-x\).  Since \(\F_p[x]\) is an
integral domain, \(g_P=AP\ne0\).  In characteristic \(p\), Frobenius gives
\[
  P(x^p)=P(x)^p,
\]
because each coefficient \(c\in\F_p\) satisfies \(c^p=c\).  Hence
\[
  f_P=P(x)^p-P(x)=P(x)\bigl(P(x)^{p-1}-1\bigr).
\]
Therefore
\[
  g_P\mid f_P
  \quad\Longleftrightarrow\quad
  AP\mid P(P^{p-1}-1).
\]
Since \(P\ne0\) and \(\F_p[x]\) is a domain, cancellation of the common factor
\(P\) is valid, so this is equivalent to
\begin{equation}\label{eq:p4e1}
    A\mid P^{p-1}-1.
\end{equation}

We next recall the needed divisibility criterion for \(A=x^p-x\).  For any
\(R\in\F_p[x]\),
\begin{equation}\label{eq:p4e2}
    A\mid R
    \quad\Longleftrightarrow\quad
    R(a)=0\text{ for all }a\in\F_p.
\end{equation}

The forward implication follows by evaluation.  Conversely, divide \(R\) by
the monic polynomial \(A\), say \(R=AQ+B\) with \(\deg B<p\).  If \(R\)
vanishes on all of \(\F_p\), then so does \(B\), since \(A(a)=0\) for every
\(a\in\F_p\).  A nonzero polynomial of degree \(<p\) over a field cannot have
all \(p\) elements of \(\F_p\) as roots, so \(B=0\), which proves \eqref{eq:p4e2}.

Apply (2) to \(R=P^{p-1}-1\).  For \(a\in\F_p\), if \(P(a)\ne0\), then
\(P(a)^{p-1}=1\); if \(P(a)=0\), then \(P(a)^{p-1}-1=-1\ne0\), including in
the case \(p=2\).  Thus \(P^{p-1}-1\) vanishes at every point of \(\F_p\)
exactly when \(P\) has no root in \(\F_p\).  Combining this with \eqref{eq:p4e1} proves
the asserted divisibility equivalence.

It remains only to record the sparse-size bounds.  If
\[
  P=\sum_{i=1}^t c_i x^{e_i},
\]
then
\[
  f_P=\sum_{i=1}^t c_i x^{pe_i}-\sum_{i=1}^t c_i x^{e_i},
  \qquad
  g_P=\sum_{i=1}^t c_i x^{p+e_i}-\sum_{i=1}^t c_i x^{e_i+1}.
\]
Thus each output has at most \(2t\) displayed terms before collecting like
terms, and collection can only decrease support.  The displayed exponents
also give \(\deg f_P\le pE\) and \(\deg g_P\le p+E\).  Multiplication and
addition of binary exponents by \(p\), together with arithmetic in
\(\F_p\), require only polynomially many bit operations in the sparse input
size.
\end{proof}

\subsection{Hardness Consequence}

We use the following finite-field sparse root-detection theorem of
Bi--Cheng--Rojas.

\begin{theorem}[Bi--Cheng--Rojas {\cite[Theorem 1.4]{BI-SIAMJC2016}}]
\label{thm:bcr}
The sparse root-detection problem
\[
  \SparseRoot
  =
  \{(p,P):p\text{ prime},\ P\in\F_p[x]\text{ sparse, and }
    P(a)=0\text{ for some }a\in\F_p\}
\]
is NP-hard under BPP many-one reductions.  More explicitly, for every
language \(L\in\NP\), there is a randomized polynomial-time map which, on
input \(y\), outputs a prime \(p_y\) and a sparse polynomial
\(P_y\in\F_{p_y}[x]\), of polynomial sparse encoding size, such that with
bounded error
\[
  y\in L
  \quad\Longleftrightarrow\quad
  P_y\text{ has a root in }\F_{p_y}.
\]
\end{theorem}

We now prove the main result of this paper.
\begin{theorem}
\label{thm-F}
\(\NONDIV\) is NP-hard under BPP many-one reductions.  Equivalently, \(\DIV\)
is hard under one-query randomized Turing reductions, whose single oracle answer
is negated.  Consequently, if sparse exact divisibility over all finite fields
had a randomized bounded-error algorithm using
\[
  \operatorname{poly}(T,\log D,\log q)
\]
bit operations on inputs \(f,g\in\F_q[x]\) with \(\deg f,\deg g<D\) and at
most \(T\) nonzero terms each, then \(\NP\subseteq\BPP\).
\end{theorem}

\begin{proof}
Let \(L\in\NP\).  By Theorem~\ref{thm:bcr}, there is a BPP many-one reduction
that maps \(y\) to a prime \(p_y\) and a sparse polynomial
\(P_y\in\F_{p_y}[x]\) such that, with bounded error,
\[
  y\in L
  \quad\Longleftrightarrow\quad
  P_y\text{ has a root in }\F_{p_y}.
\]
Apply Lemma~\ref{lem:bridge} to form the sparse pair
\((f_{P_y},g_{P_y})\).  The lemma gives
\[
  P_y\text{ has a root in }\F_{p_y}
  \quad\Longleftrightarrow\quad
  g_{P_y}\nmid f_{P_y}.
\]
Therefore, the randomized map
\[
  y\longmapsto (p_y,f_{P_y},g_{P_y})
\]
is a BPP many-one reduction from \(L\) to \(\NONDIV\).  The output size is
polynomial because the source reduction has polynomial sparse output size and
Lemma~\ref{lem:bridge} increases support by at most a constant factor.  More
precisely, if \(P_y\ne0\) and \(\deg P_y<D\), and if \(D=1\) is used for the
guarded case \(P_y=0\), then the target degree bound may be taken to be
\[
  D' = 1+\max\{p_yD,\ p_y+D,\ 1\},
\]
so \(\log D'=O(\log p_y+\log(D+1))\). This proves the BPP many-one NP-hardness
of \(\NONDIV\).

An oracle for \(\DIV\) decides \(\NONDIV\) on the same well-formed domain by
negating its answer. Hence, the same construction gives a one-query randomized
Turing reduction to \(\DIV\) with answer negation.

Finally, suppose that there is a randomized bounded-error polynomial-time
algorithm for \(\DIV\) over all finite fields, with running time polynomial in
\(T,\log D,\log q\). It applies, in particular, to the prime-field instances
constructed above.  Negating its answer gives a randomized bounded-error
algorithm for \(\NONDIV\) on those instances, and composing it with the BPP
many-one reduction just proved decides every \(L\in\NP\) in BPP.  Standard
amplification keeps the total error bounded. Therefore, \(\NP\subseteq\BPP\).
\end{proof}

The theorem deliberately does not assert the ordinary yes-to-yes Karp
NP-hardness of \(\DIV\).  The bridge has the opposite polarity:
\(\DIV\)-yes instances correspond to root nonexistence, while
\(\NONDIV\)-yes instances correspond to root existence.

\section{Conclusion}
\label{sec-conc}
We consider the problem of divisibility for sparse polynomials with coefficients in a finite field, namely determining whether one sparse polynomial is exactly divisible by another.
It is shown that this problem is CoNP-hard; that is, deciding whether a sparse polynomial does not divide another sparse polynomial exactly is NP-hard under standard complexity assumptions.

\section*{Acknowledgment}
This paper is supported by the Strategic Priority Research Program of CAS Grants XDA0480502 and XDA0480503, NSFC Grants 12288201 and 92270001.

\bibliographystyle{KLMM/klmm}   
\bibliography{Refs} 

\appendix
\section{A Randomized Polynomial Reduction from Integer Non-Divisibility to Finite Field Non-Divisibility}
\label{sec-red}

In this section, we present an alternative proof of Theorem \ref{thm-main} by constructing a randomized polynomial-time reduction that transforms integer non-divisibility into non-divisibility over a finite field.
Define the decision problem

\begin{definition}[Non-Divisibility-$\R$]
Decide whether a sparse polynomial $g\in \R[x]$ does not divide another sparse polynomial $f\in \R[x]$ for a ring $\R$.
\end{definition}
The reduction is a randomized polynomial-time Turing reduction to the finite-field non-divisibility oracle.  In branches where the integer-side answer has already been determined, the reduction may return directly; equivalently, those branches may be encoded by fixed sentinel oracle instances over \(\mathbb F_2\),
for instance \((0,1)\) for a YES instance and \((1,0)\) for a NO instance.

\begin{theorem}
\label{thm-Red-I2F}
Under this randomized Turing/oracle reduction reading, Non-Divisibility-$\Z$ reduces in randomized polynomial time to Non-Divisibility-$\F_q$.
\end{theorem}

Since Non-Divisibility-$\Z$ is proved to be NP-complete in \cite{plaisted1984}, Theorem \ref{thm-main} derives from Theorem \ref{thm-Red-I2F}.

Theorem \ref{thm-Red-I2F} is proved in three steps.
In Step 1 presented in Section \ref{sec-red1}, some preliminary cases are considered.
In Step 2 presented in Section \ref{sec-red2}, the main reduction step is given.
In Step 3 presented in Section \ref{sec-red3}, the final proof of Theorem \ref{thm-Red-I2F} is presented.

\subsection{Preliminary Reductions}
\label{sec-red1}
Let \((g,f)\in\mathbb Z[x]^2\) be a sparse input pair, and let \(S\) be its
sparse binary input length.  The first step combines equal exponents and
deletes zero coefficients.  This normalization is polynomial time in \(S\), and
we assume it has been performed.

The following boundary cases are decided without using the main randomized
step.  If \(g=0\), then \(0\mid f\) if and only if \(f=0\); hence the
non-divisibility answer is NO when \(f=0\) and YES otherwise.  If
\(g\neq 0\) and \(f=0\), then \(f=g\cdot 0\), so the answer is NO.  If
\(g=c\) is a nonzero constant, then \(c\mid f\) in \(\mathbb Z[x]\) if and
only if \(c\) divides every coefficient of \(f\), with missing sparse
coefficients treated as zero; this is checked by integer divisibility tests.
Finally, after excluding \(f=0\), if \(0<\deg g>\deg f\), then \(g\) cannot
divide \(f\), because \(\deg(gh)=\deg g+\deg h\) for nonzero
\(h\in\mathbb Z[x]\).  This branch is therefore a YES instance.

It remains to consider the case
\[
        f,g\neq 0,\qquad m=\deg g>0,\qquad m\leq \deg f .
\]
Let
\[
        c=\cont(g)>0,\qquad G=g/c .
\]
Here \(c\) is the gcd of the absolute values of the listed nonzero
coefficients of \(g\), and \(G\in\mathbb Z[x]\) is primitive with
\(\deg G=m\).  Both \(c\) and \(G\) are computable in polynomial time from the
sparse input.  The literal divisibility problem in \(\mathbb Z[x]\) splits as
\[
        cG\mid f
        \quad\Longleftrightarrow\quad
        \bigl(c\text{ divides every coefficient of }f\bigr)
        \text{ and } G\mid f/c \text{ in }\mathbb Z[x].
\]
Indeed, if \(f=cGH\), then all coefficients of \(f\) are multiples of \(c\)
and \(f/c=GH\).  Conversely, if all coefficients of \(f\) are divisible by
\(c\) and \(f/c=GH\), then \(f=cGH\).  Hence, if some coefficient of \(f\) is
not divisible by \(c\), the reduction returns YES.  Otherwise set
\[
        F=f/c\in\mathbb Z[x].
\]
Since \(G\) is primitive, Gauss's lemma in primitive divisor form gives
\[
        G\mid F \text{ in }\mathbb Z[x]
        \quad\Longleftrightarrow\quad
        G\mid F \text{ in }\mathbb Q[x].
\]

\subsection{The Pseudo-Remainder Mechanism}
\label{sec-red2}
The randomized part of the reduction chooses a large prime and asks the finite-field oracle about the reduction of the primitive pair \((G,F)\).  The
key point is that, when \(G\nmid F\), all primes that incorrectly make
\(G\bmod p\) divide \(F\bmod p\) are forced to divide one explicitly bounded
integer arising from pseudo-division.

\begin{lemma}
\label{lem-red1}
Let \(G,F\in\mathbb Z[x]\), let \(G\) be primitive, let
\(m=\deg G>0\), let \(n=\deg F\), and set \(a=\lc(G)\).  Suppose
\(G\nmid F\) in \(\mathbb Q[x]\).  Then there is a nonzero polynomial
\(P\in\mathbb Z[x]\), with \(\deg P<m\), such that every prime \(p\) for which
\((G\bmod p)\mid(F\bmod p)\) in \(\mathbb F_p[x]\) divides
\[
        A=a\,\cont(P).
\]
Moreover, if \(D=\max(m,n,1)\) and the absolute values of all coefficients of
\(G\) and \(F\) are at most \(H=2^B\), then the bit length of \(A\) is bounded
by \(C_0D(S+1)\) for an absolute constant \(C_0\), where \(S\) is the sparse
input length.
\end{lemma}

\begin{proof}
Starting from \(R_0=F\), perform the pseudo-division recurrence over
\(\mathbb Z[x]\).  While \(d_j=\deg R_j\geq m\), with
\(b_j=\lc(R_j)\), set
\[
        R_{j+1}=aR_j-b_jx^{d_j-m}G .
\]
The leading terms cancel, so the degree strictly decreases at each step.  The
process stops after some
\[
        e\leq \max(n-m+1,0)\leq D+1
\]
steps with a polynomial \(P\in\mathbb Z[x]\) satisfying \(\deg P<m\), and
induction on the recurrence gives
\[
        a^eF=QG+P
\]
for some \(Q\in\mathbb Z[x]\).  Over \(\mathbb Q[x]\), the ordinary remainder
of \(F\) upon division by \(G\) is \(a^{-e}P\).  Since \(G\nmid F\) in
\(\mathbb Q[x]\), this remainder is nonzero, and therefore \(P\neq 0\).

Let \(p\nmid A=a\,\cont(P)\).  Then \(p\nmid a\), so \(G\bmod p\) still has
degree \(m\).  Also \(p\nmid\cont(P)\), so \(P\bmod p\) is a nonzero
polynomial of degree \(<m\).  Reducing
\[
        a^eF=QG+P
\]
modulo \(p\) gives
\[
        (a\bmod p)^e(F\bmod p)=(Q\bmod p)(G\bmod p)+(P\bmod p).
\]
If \(G\bmod p\) divided \(F\bmod p\), then it would divide
\((a\bmod p)^e(F\bmod p)\), and hence it would divide \(P\bmod p\).  This is
impossible because \(G\bmod p\) has positive degree \(m\), while
\(P\bmod p\) is nonzero of degree less than \(m\).  Thus every prime producing
the false divisibility relation divides \(A\).

For the size bound, let \(M_j\) be the maximum absolute coefficient of \(R_j\).
Since \(M_0\leq H\), \(|a|\leq H\), and every coefficient of \(G\) has absolute
value at most \(H\), the recurrence gives
\[
        M_{j+1}\leq 2HM_j .
\]
Consequently
\[
        M_e\leq (2H)^{D+1}H .
\]
Every coefficient of \(P=R_e\), and hence \(\cont(P)\), has bit length at most
\[
        1+(D+2)(B+1).
\]
Together with \(|a|\leq H\), this bounds the bit length of
\(A=a\,\cont(P)\) by \(C_0D(S+1)\) for an absolute constant \(C_0\).
\end{proof}

\subsection{The Randomized Oracle Reduction}\label{sec-red3}

We now prove Theorem \ref{thm-Red-I2F}.
\begin{proof}
We now complete the construction in the remaining primitive branch.  Let
\[
        n=\deg F,\qquad D=\max(m,n,1).
\]
Since exponents are binary encoded, \(D\leq 2^S\). 
Lemma \ref{lem-red1} gives constants
\(C_2,a_0\) such that, whenever \(G\nmid F\),
\[
        \log_2 |A|\leq 2^{C_2S^{a_0}} .
\]
A nonzero integer of bit length \(L\) has at most \(L\) distinct prime
divisors, so the number of bad primes, namely primes \(p\) for which
\((G\bmod p)\mid(F\bmod p)\) despite \(G\nmid F\), is at most
\[
        2^{C_2S^{a_0}} .
\]

Choose fixed constants \(a_1>a_0\) and \(C_1\) sufficiently large, and set
\[
        K=\lceil C_1S^{a_1}\rceil .
\]
The reduction samples uniformly from the primes in the dyadic interval
\[
        [2^{K-1},2^K).
\]
It does so by rejection sampling from that interval,
using a polynomial-time primality test.  Conditioned on acceptance, the output
prime is uniform among the primes in the interval, and the expected number of
trials is polynomial in \(K\).  By the prime number theorem, for all
sufficiently large \(K\) there are at least \(2^{K-2}/K\) primes in this
interval; increasing \(C_1\), or hardwiring the finitely many smaller input
sizes, handles the finite exceptions.  Therefore, in the nondivisible
primitive branch,
\[
        \Pr[p\text{ is bad}]
        \leq K\,2^{C_2S^{a_0}+2-K}
        \leq \frac13 .
\]
The sampled prime has \(K=\operatorname{poly}(S)\) bits.

The oracle query is the sparse finite-field instance
\[
        (G_p,F_p)=(G\bmod p,F\bmod p)\in\mathbb F_p[x]^2 .
\]
The field is represented by the \(K\)-bit prime \(p\).  Coefficient residues
have \(O(K)\) bits, exponents are inherited from the input, and terms whose
coefficients reduce to zero are deleted.  Thus the oracle query has size
polynomial in \(S\).

We verify the two correctness directions for the reduction.  First suppose
that the integer instance is a NO instance of Non-Divisibility-I, so
\(g\mid f\) in \(\mathbb Z[x]\).  The boundary and constant branches return
NO by the equivalences proved above.  In the remaining branch, the content
split gives \(F=GH\) for some \(H\in\mathbb Z[x]\).  Reducing modulo any prime
\(p\) gives
\[
        F_p=G_pH_p .
\]
Because \(G\) is primitive, \(G_p\) is not the zero polynomial for any prime
\(p\); otherwise \(p\) would divide every coefficient of \(G\).  Hence every
oracle query in this branch is a NO instance of Non-Divisibility-F.  The
reduction therefore returns NO with probability \(1\).

Now suppose that the integer instance is a YES instance, so \(g\nmid f\) in
\(\mathbb Z[x]\).  All direct YES branches have already been justified.  In
the remaining branch, \(F=f/c\) is defined and the content split gives
\(G\nmid F\) in \(\mathbb Z[x]\).  Since \(G\) is primitive, Gauss's lemma
gives \(G\nmid F\) in \(\mathbb Q[x]\).  The pseudo-remainder lemma shows that
all primes for which \(G_p\mid F_p\) divide the controlling integer
\(A=a\,\cont(P)\), and the large-prime sampler avoids these primes with
probability at least \(2/3\).  With that probability the oracle query is a YES
instance of Non-Divisibility-F, and the reduction returns YES.

All computations performed by the oracle machine are randomized polynomial
time in \(S\): normalization, degree comparison, integer gcd and divisibility
tests, coefficient division by \(c\), sampling a \(K\)-bit prime, reducing
coefficients modulo \(p\), and writing the sparse oracle query.  The
pseudo-remainder \(P\) and the controlling integer \(A\) are used only in the
analysis and are not computed by the reduction.  The reduction has one-sided
error.  Repeating the randomized trial independently and accepting YES if any
trial returns YES amplifies the success probability.  This proves the
theorem.
\end{proof}

\end{document}